\documentclass{article}
\usepackage{spconf,amsmath,graphicx,hyperref}

\usepackage{amssymb}    
\usepackage{amsthm}     
\usepackage{mathrsfs}   

\usepackage{enumerate}

\usepackage{subcaption}

\usepackage{hyperref}
\usepackage[nameinlink]{cleveref}
\Crefname{figure}{Fig.}{Figs.}
\Crefname{section}{Sec.}{Secs.}
\Crefname{table}{Tab.}{Tabs.}
\Crefname{equation}{Eq.}{Eqs.}

\newtheorem{theorem}{Theorem}

\newtheorem{proposition}{Proposition}

\title{On Sampling of Multiple Correlated Stochastic Signals}
\name{Lin Jin$^{a}$ \qquad Hang Sheng$^{a}$ \qquad Hui Feng$^{a,b,\star}$ \qquad Bo Hu$^{a,b}$ \thanks{$^{\star}$ Corresponding author.}}
\address{$^{a}$ College of Future Information Technology, Fudan University, Shanghai 200433, China \\
    $^{b}$ State Key Laboratory of Integrated Chips and Systems, Fudan University, Shanghai 200433, China}

\begin{document}
\ninept
\maketitle
\begin{abstract}
Multiple stochastic signals possess inherent statistical correlations, yet conventional sampling methods that process each channel independently result in data redundancy.
To leverage this correlation for efficient sampling, we model correlated channels as a linear combination of a smaller set of uncorrelated, wide-sense stationary latent sources.
We establish a theoretical lower bound on the total sampling density for zero mean-square error reconstruction, proving it equals the ratio of the total spectral bandwidth of latent sources to the number of correlated signals.
We then develop a constructive multi-band sampling scheme to achieve this bound. 
The proposed method operates via spectral partitioning of the latent sources, followed by spatio-temporal sampling and interpolation. 
Experiments on synthetic and real datasets confirm that our scheme achieves near-lossless reconstruction precisely at the theoretical sampling density, validating its efficiency.
\end{abstract}
\section{Introduction}
\label{sec:intro}

Multi-channel stochastic signals, which we term multiple stochastic signals in this paper, arise across many modern signal-processing applications and often serve as the primary information-bearing quantities in various complex systems.
Notably, these channels are frequently not independent but instead inherently exhibit statistical correlation.
For instance, in Multiple-Input Multiple-Output (MIMO) wireless receivers, spatial and propagation effects induce inter-channel correlation \cite{bjornson2017massive};
in multiple biomedical recordings, neural synchrony produce correlated activity \cite{panzeri2022structures}; 
and in multi-microphone audio captures, a single acoustic source creates correlated recordings across sensors \cite{blanco2020microphone}.
The structure embedded in these correlations contains system-level information unavailable from any single channel.
Leveraging such correlation therefore becomes critical to increase signal processing efficiency.

The proliferation of multiple systems leads to the generation of massive datasets, imposing heavy burdens on communication, storage, and computational resources \cite{hashem2015rise}.
Efficiently converting these continuous-time signals into discrete observations through sampling is therefore a critical necessity for mitigating resource costs. 
In multiple scenarios, naively sampling each channel independently ignores the information reuse potential offered by inter-channel correlation, leading to redundant data acquisition.
Consequently, developing sampling theories that explicitly exploit correlation to minimize sampling cost is of paramount importance.

For single-channel signals, the classical Shannon-Nyquist sampling theorem is a cornerstone for deterministic signals.
It is less well known that the theorem can also extend to wide-sense stationary (WSS) stochastic signals with band-limited power spectral density (PSD) \cite{papoulis2002probability}.
Existing studies have investigated sampling of single stochastic signals  \cite{rosenblatt2012random,kipnis2017distortion}, but they do not address the multiple signals.

As to multiple stochastic signals, research has explored correlation from various perspectives.
Shlezinger et al. \cite{shlezinger2019joint} proposed a distributed sampling and joint reconstruction framework, demonstrating the benefit of leveraging correlations; nevertheless, a theoretical analysis remains lacking. 
Theories rooted in compressed sensing have partially addressed rate reduction, yet their guarantees rely on strong sparsity assumptions, limiting their applicability to more general signals \cite{mishali2010theory,angrisani2014multi,amini2021exploiting}.
Based on a latent source model, Ahmed and Romberg \cite{ahmed2019compressive,ahmed2023sub} proposed a random mixing architecture with corresponding sampling rate bound. However, their approach requires a dedicated analog front-end to physically mix signals before sampling and targets deterministic waveform recovery rather than stochastic sampling guarantees.
In summary, there is a lack of comprehensive theoretical guarantees on the sampling lower bound for multiple correlated stochastic signals, and practical schemes that attain this bound are still less explored.

We model the correlated signals as linear combinations of low-dimensional latent sources and establish a lower bound on the total sampling density required for reconstruction with zero mean squared error (MSE).
This result reveals the connection between signal correlation and sampling efficiency.
Building upon this bound, we then propose a constructive multi-band sampling scheme that achieves the minimum rate via power spectrum partitioning.

The main contributions of this paper are:
(i) Under the latent source generation model, 
we prove that the lower bound on total sampling density for reconstruction with zero MSE equals the ratio of the latent signals' total spectral bandwidth to the number of correlated signals;
(ii) We propose a constructive multi-band sampling scheme to achieve the lowest total sampling density while ensuring reconstruction in the mean-square sense;
(iii) We validate the effectiveness of the proposed scheme through experiments on both synthetic and real datasets, demonstrating that the scheme realizes near-lossless reconstruction at the theoretical sampling density.

\section{signal model}
\label{sec:signal model}

\begin{figure}[t]
\centering
\includegraphics[width=8.2cm]{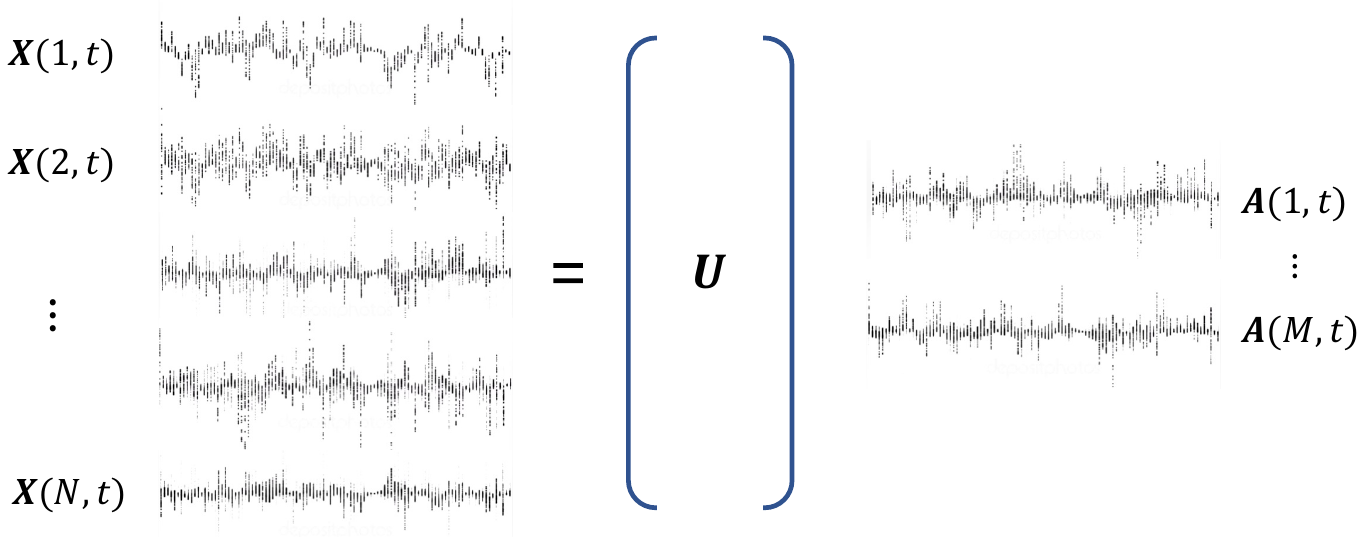}
\caption{Correlated signal model.}
\label{fig:correlated signal}
\end{figure}

Let $\boldsymbol{\mathit{X}}=[X(1,t),X(2,t),...,X(N,t)]^T,t\in \mathbb{R}$ denote $N$-channel correlated stochastic signals, where each channel $X(n,\cdot),n=1,...,N$ is a continuous-time signal.
We establish a correlated signal model inspired by \cite{ahmed2019compressive}, in which correlated signals can be represented as linear combinations of a smaller number of latent signals.
Specifically, we assume that each channel $X(n,\cdot)$ is a linear mixture of $M$ mutually uncorrelated signals $\boldsymbol{\mathit{A}}=[A(1,t),...,A(M,t)]^T,t\in \mathbb{R}$:
\begin{equation}
    \label{eq:signal_model}
    \boldsymbol{\mathit{X}}=\boldsymbol{\mathit{U}}\boldsymbol{\mathit{A}}.
\end{equation}
Each channel of the latent signals $\boldsymbol{\mathit{A}}$ is a WSS continuous-time signal and has finite power spectral bandwidth.
The mixing matrix $\boldsymbol{\mathit{U}}\in \mathbb{R}^{N\times M}$ combines the $M$ latent signals.
The signal model is illustrated in \Cref{fig:correlated signal}.

The generative model can be interpreted from two perspectives.
The first is a physical channel-mixing interpretation, common in applications like wireless communications or array processing, where physically distinct sources $\boldsymbol{\mathit{A}}$ are mixed by a channel $\boldsymbol{\mathit{U}}$ to induce the observed correlation $\boldsymbol{\mathit{X}}$.
The second is a statistical latent-factor interpretation, analogous to Principal Component Analysis (PCA) or Factor Analysis (FA), where $\boldsymbol{\mathit{A}}$ represents abstract uncorrelated factors that reveal an underlying low-dimensional structure in $\boldsymbol{\mathit{X}}$ and are not required to be physically separable.
While their conceptual origins differ, both interpretations converge on the same algebraic structure presented in \Cref{eq:signal_model}.

Denote the correlation function of $\boldsymbol{\mathit{X}}$ by
\begin{equation*}
    \boldsymbol{\mathit{R}}_{\boldsymbol{\mathit{X}}}(i,j,\tau)=\mathbb{E}\left [ X(i,t)X^H(j,t-\tau)\right ],
\end{equation*}
where it reduces to the autocorrelation of a single channel when $i=j$, and to the cross-correlation between two channels when $i\ne j$.
If there exist indices $i\ne j$ and lags $\tau$ such that $\boldsymbol{\mathit{R}}_{\boldsymbol{\mathit{X}}}(i,j,\tau) \ne 0$, then the signals $X(i,\cdot)$ and $X(j,\cdot)$ exhibit inter-channel statistical correlation. 
In this case, the multiple stochastic signals $\boldsymbol{\mathit{X}}$ is referred to as correlated.
Correspondingly, we define the latent stochastic signals $\boldsymbol{\mathit{A}}$ as uncorrelated if the cross-correlation function between any two channels is identically zero:
\begin{equation}
    \label{eq:uncorrelate definition}
    \boldsymbol{\mathit{R}}_{\boldsymbol{\mathit{A}}}(i,j,\tau) \equiv  0,\quad \forall i\ne j,\, \forall \tau \in \mathbb{R}.
\end{equation}


By the Wiener-Khinchin theorem, the auto-(or cross-) power spectral density of $\boldsymbol{\mathit{X}}$ is the Fourier transform of its auto-(or cross-) correlation function \cite{kay1988modern}: $\boldsymbol{\mathit{S}}_{\boldsymbol{\mathit{X}}}(i,j,f)= \mathscr{F}\{\boldsymbol{\mathit{R}}_{\boldsymbol{\mathit{X}}}(i,j,\tau)\}.$
For brevity, we simplify the notation by retaining only the autocorrelation of $\boldsymbol{\mathit{A}}$, denote as $\boldsymbol{\mathit{R}}_{\boldsymbol{\mathit{A}}}(m,\tau)$ for $m=1,...,M,\, \tau \in \mathbb{R}$, with the corresponding power spectral densities denoted by $\boldsymbol{\mathit{S}}_{\boldsymbol{\mathit{A}}}(m,f)$.

Assume the power spectrum supports of $\boldsymbol{\mathit{A}}$ are known.
For each channel $m$, let the spectral support set $\mathcal{B}_m \subset \mathbb{R}$ be Lebesgue-measurable with finite measure, i.e. $\boldsymbol{\mathit{S}}_{\boldsymbol{\mathit{A}}}(m,f)=0$ for $f \notin \mathcal{B}_m$.
Denote by $\mu(\cdot)$ the Lebesgue measure and define the total power spectrum bandwidth of $\boldsymbol{\mathit{A}}$ as
\begin{equation}
    \label{eq:bandlimited_spectrum}
    B:=\sum_{m=1}^{M} \mu(\mathcal{B}_m),
\end{equation}
with $B<+\infty$.

For the $n$-th channel signal $X(n,\cdot)$, its temporal sampling set is denoted as $\mathcal{S}_n=\{t_{ns}:t_{ns} \in \mathbb{R}\}$.
The overall sampling set for multiple signals $\boldsymbol{\mathit{X}}$ is 
\begin{equation*}
    \mathcal{S}=\{(n_s,t_{ns}):n_s\in \{1,...,N\},t_{ns} \in \mathbb{R}\}.
\end{equation*}
We define its spatial projection as $\mathcal{S}_C=\{n_s:(n_s,t_{ns}) \in \mathcal{S}\}$ and the temporal projection as $\mathcal{S}_T=\sum_{n_s\in \mathcal{S}_C}  \mathcal{S}_n=\{t_{ns}:(n_s,t_{ns}) \in \mathcal{S}\}$.


To quantify the overall sampling cost of the $N$-channel signals $\boldsymbol{\mathit{X}}$, particularly for non-uniform sampling schemes, the total average sampling density is defined in \cite{ji2019hilbert} as
\begin{equation*}
    D(\mathcal{S}):=\liminf_{t \to \infty}  \frac{|\mathcal{S} \cap \{\{1,...,N\}\times[-t, t]\}|}{2tN}.
\end{equation*}


Focus on multiple correlated stochastic signals described above, our objective is to develop a sampling theory that derives a lower bound of $D(\mathcal{S})$ and to construct a feasible sampling and reconstruction scheme.

\section{Sampling theorem}
\label{sec:sampling theorem}

Based on the signal model in \Cref{sec:signal model}, we present a sampling theorem for multiple correlated stochastic signals that operates directly on the $N$ channels, without any pre-sampling analog mixing stages.
The key principle is that the information rate required to capture the $N$ correlated signals is governed by the spectral support of the $M$ underlying, uncorrelated latent signals. 
Hence the sampling density of $\boldsymbol{\mathit{X}}$ can be related to the total spectral bandwidth of $\boldsymbol{\mathit{A}}$. 
\Cref{ssec:main result} states the main result and sketches its proof;
\Cref{ssec:multi-band sampling scheme} then gives a constructive multi-band sampling and reconstruction scheme to achieve the bound.

\subsection{Main result}
\label{ssec:main result}

Given the generative signal model \eqref{eq:signal_model}, we now first specify the conditions under which $\boldsymbol{\mathit{X}}$ is guaranteed to be correlated. 

\begin{proposition}
Let the signals $\boldsymbol{\mathit{X}}$ be generated by the model in \eqref{eq:signal_model}, $\boldsymbol{\mathit{X}}$ is guaranteed to be correlated if the following conditions hold:
(i) The linear transformation is dimensionality-increasing, i.e., $N > M$;
(ii) Each $A(m,\cdot)$ has positive power, i.e., $\mathbb{E}\left [\left |A(m,\cdot)\right |^2\right ] > 0$ for all $m\in \{1,...,M\}$;
(iii) Each row of the mixing matrix $\boldsymbol{\mathit{U}}$ is a nonzero vector.
\end{proposition}

\begin{proof}
Two mild conditions (ii) and (iii) are imposed to exclude degenerate cases where an mixed channel would be identically zero.

For $i,j\in \{1,...,N\},\, \tau\in \mathbb{R}$, we have
\begin{equation}
	\begin{split}
    \boldsymbol{\mathit{R}}_{\boldsymbol{\mathit{X}}}(i,j,\tau)&=\mathbb{E}[X(i,t) X^H(j,t-\tau)] \notag \\
    &=\mathbb{E}[U(i,\cdot)A(\cdot,t) A^H(\cdot,t-\tau)U^H(j,\cdot)] \notag \\
    &=U(i,\cdot) \mathbb{E}[A(\cdot,t) A^H(\cdot,t-\tau)]U^H(j,\cdot) \notag \\
    &=U(i,\cdot)\boldsymbol{\mathit{R}}_{\boldsymbol{\mathit{A}}}(\cdot,\cdot,\tau)U^H(j,\cdot).
	\end{split}
\end{equation}

If the time difference $\tau$ is fixed, we can see $\boldsymbol{\mathit{R}}_{\boldsymbol{\mathit{X}}}(i,j,\tau)$ as a $N\times N$ matrix $\boldsymbol{\mathit{R}}_{\boldsymbol{\mathit{X}}}$, and $\boldsymbol{\mathit{R}}_{\boldsymbol{\mathit{A}}}$ is a $M\times M$ diagonal matrix,  $\boldsymbol{\mathit{R}}_{\boldsymbol{\mathit{A}}}=\text{diag}(\boldsymbol{\mathit{R}}_{\boldsymbol{\mathit{X}}}(1,1,\tau),\boldsymbol{\mathit{R}}_{\boldsymbol{\mathit{X}}}(2,2,\tau),...,\boldsymbol{\mathit{R}}_{\boldsymbol{\mathit{X}}}(M,M,\tau))$.
For a fixed time difference, the relationship of correlation functions between signals $\boldsymbol{\mathit{X}}$ and $\boldsymbol{\mathit{A}}$ can be expressed using matrix operations: $$\boldsymbol{\mathit{R}}_{\boldsymbol{\mathit{X}}}=\boldsymbol{\mathit{U}} \boldsymbol{\mathit{R}}_{\boldsymbol{\mathit{A}}} \boldsymbol{\mathit{U}}^H,$$
The entry of row $i$, column $j$ is:
\begin{equation}
	\begin{split}
    \boldsymbol{\mathit{R}}_{\boldsymbol{\mathit{X}}}(i,j,\tau)&=[\boldsymbol{\mathit{U}} \boldsymbol{\mathit{R}}_{\boldsymbol{\mathit{A}}} \boldsymbol{\mathit{U}}^H]_{ij} \notag \\
    &=\sum_{k=1}^{M} \sum_{l=1}^{M} \boldsymbol{\mathit{U}}(i,k) \boldsymbol{\mathit{R}}_{\boldsymbol{\mathit{A}}}(k,l,\tau) \overline{\boldsymbol{\mathit{U}}(j,l)},
	\end{split}
\end{equation}
and in addition, from $\boldsymbol{\mathit{A}}$ is uncorrelated we know $\boldsymbol{\mathit{R}}_{\boldsymbol{\mathit{A}}}(k,l,\tau)=0$ for $k\ne l$, so
\begin{equation}
	\begin{split}
    \boldsymbol{\mathit{R}}_{\boldsymbol{\mathit{X}}}(i,j,\tau)
    &=\sum_{k=1}^{M} \sum_{l=1}^{M} \boldsymbol{\mathit{U}}(i,k) \boldsymbol{\mathit{R}}_{\boldsymbol{\mathit{A}}}(k,l,\tau) \overline{\boldsymbol{\mathit{U}}(j,l)} \notag \\
    &=\sum_{k=1}^{M} \boldsymbol{\mathit{U}}(i,k) \boldsymbol{\mathit{R}}_{\boldsymbol{\mathit{A}}}(k,k,\tau) \overline{\boldsymbol{\mathit{U}}(j,k)} \notag \\
    &=\sum_{k=1}^{M} \boldsymbol{\mathit{U}}(i,k)  \overline{\boldsymbol{\mathit{U}}(j,k)} \boldsymbol{\mathit{R}}_{\boldsymbol{\mathit{A}}}(k,k,\tau).
	\end{split}
\end{equation}

Based on the preliminaries above, we now proceed with a proof by contradiction.
The core idea is to assume $\boldsymbol{\mathit{X}}$ are uncorrelated and show that this assumption, when combined with the model's given conditions, leads to a logical impossibility.

Assume that $\boldsymbol{\mathit{X}}$ is uncorrelated, which means $\forall i,j\in \{1,...,N\}$, $i\ne j$ and $\forall \tau$, we have 
\begin{equation}
\label{eq:lemma1-1}
    \boldsymbol{\mathit{R}}_{\boldsymbol{\mathit{X}}}(i,j,\tau)=0.
\end{equation}

Consider the case where $\tau =0$. Denote $\boldsymbol{\mathit{R}}_{\boldsymbol{\mathit{A}}}(k,k,0)=\mathbb{E}\left [\left |A(k,\cdot)\right |^2\right ]$ as $d_k$, and $\boldsymbol{\mathit{R}}_{\boldsymbol{\mathit{A}}}(\cdot,\cdot,0)\allowbreak =\text{diag}(d_1, d_2,..., d_M)$.
Then \eqref{eq:lemma1-1} yields 
\begin{equation}
\label{eq:lemma1-2}
    \sum_{k=1}^{M} \boldsymbol{\mathit{U}}(i,k)  \overline{\boldsymbol{\mathit{U}}(j,k)} d_k=0 \quad \text{for} \; i\ne j.
\end{equation}

Let $\boldsymbol{\mathit{v}}_i \in \mathbb{C}^{1\times M}$ be the row vector of $\boldsymbol{\mathit{U}}$, $\boldsymbol{\mathit{v}}_i=[\boldsymbol{\mathit{U}}(i,1),\allowbreak \boldsymbol{\mathit{U}}(i,2),...,\boldsymbol{\mathit{U}}(i,M)]$.
From condition (ii) we know $d_k > 0$, therefore $\boldsymbol{\mathit{R}}_{\boldsymbol{\mathit{A}}}(\cdot,\cdot,0)$ is a positive-definite matrix.
Let $\boldsymbol{\mathit{G}}=\sqrt{\boldsymbol{\mathit{R}}_{\boldsymbol{\mathit{A}}}(\cdot,\cdot,0)}=\text{diag}(\sqrt{d_1},\allowbreak  \sqrt{d_2},...,\sqrt{d_M})$, and we have $\boldsymbol{\mathit{R}}_{\boldsymbol{\mathit{A}}}(\cdot,\cdot,0) =\boldsymbol{\mathit{G}}^H \boldsymbol{\mathit{G}}=\boldsymbol{\mathit{G}} \boldsymbol{\mathit{G}}$.
Then \eqref{eq:lemma1-2} yields
\begin{equation}
\label{eq:lemma1-3}
	\begin{split}
    \boldsymbol{\mathit{v}}_i \boldsymbol{\mathit{R}}_{\boldsymbol{\mathit{A}}}(\cdot,\cdot,0) \boldsymbol{\mathit{v}}_j^H 
    &= \boldsymbol{\mathit{v}}_i \boldsymbol{\mathit{G}}^H \boldsymbol{\mathit{G}} \boldsymbol{\mathit{v}}_j^H \\
    &= (\boldsymbol{\mathit{G}} \boldsymbol{\mathit{v}}_i^H)^H (\boldsymbol{\mathit{G}} \boldsymbol{\mathit{v}}_j^H) \\
    &=0  \qquad \qquad \qquad \qquad \text{for} \; i\ne j.
	\end{split}
\end{equation}

Let $\boldsymbol{\mathit{y}}_i=\boldsymbol{\mathit{G}} \boldsymbol{\mathit{v}}_i^H \in \mathbb{C}^M,\, i\in \{1,...,N\}$, then \eqref{eq:lemma1-3} yields
\begin{equation}
\label{eq:lemma1-4}
    \boldsymbol{\mathit{y}}_i^H \boldsymbol{\mathit{y}}_j =0  \quad \text{for} \; i\ne j.
\end{equation}
which means the set $\{\boldsymbol{\mathit{y}}_1,\boldsymbol{\mathit{y}}_2,...\boldsymbol{\mathit{y}}_N\}$ is an orthogonal set of vectors.
Since $\boldsymbol{\mathit{y}}_i=\vec{0}$ if and only if $\sqrt{d_k} \, \overline{\boldsymbol{\mathit{U}}(i,k)}=0$ for $\forall k\in \{1,...,M\}$,  
and from condition (ii) and (iii) we know that $\sqrt{d_k} \, \overline{\boldsymbol{\mathit{U}}(i,k)}>0$, thus $\boldsymbol{\mathit{y}}_i\ne \vec{0}$.

Therefore, $\{\boldsymbol{\mathit{y}}_1,\boldsymbol{\mathit{y}}_2,...\boldsymbol{\mathit{y}}_N\}$ is an orthogonal set containing $N$ nonzero vectors.
A set of $N$ nonzero orthogonal vectors is necessarily linearly independent. 
In an $M$-dimensional vector space, the number of linearly independent vectors cannot exceed the dimension $M$.
Thus, we must have $N \le M$.
However, this contradicts our initial premise (i), therefore the assumption cannot be true.
Hence, $\boldsymbol{\mathit{X}}$ is correlated.
\end{proof}

\begin{theorem}
Let $\boldsymbol{\mathit{A}}$ be multiple uncorrelated stochastic signals with total spectral bandwidth $B<+\infty$, as defined in \eqref{eq:bandlimited_spectrum}. 
Let $\boldsymbol{\mathit{U}}\in \mathbb{R}^{N\times M}$ be a mixing matrix with 
$\text{rank}(\boldsymbol{\mathit{U}})=M$
that generates correlated signals $\boldsymbol{\mathit{X}}$ according to \eqref{eq:signal_model}.
Then, there exists a sampling set $\mathcal{S}$ satisfies
\begin{equation*}
    D(\mathcal{S})\ge \frac{B}{N}
\end{equation*}
such that a reconstruction $\hat{\boldsymbol{\mathit{X}}}$ of the signals $\boldsymbol{\mathit{X}}$ can be obtained, satisfying the zero MSE condition: $\mathbb{E}\left [ \left |\boldsymbol{\mathit{X}}-\hat{\boldsymbol{\mathit{X}}}\right |^2\right ]=0$.
\end{theorem}

\begin{proof}
As to latent signals $\boldsymbol{\mathit{A}}$, for single channel $m$, partition $B_m$ into $L_m$ disjoint subbands $\mathcal{B}_m=\bigcup_{l=1}^{L_m}\mathcal{B}_{m,l}$, where each subband has bandwidth $B_{m,l}=\mu(\mathcal{B}_{m,l})$ and $\sum_{l=1}^{L}B_{m,l}=B_m$.

Apply bandpass filtering on the original power spectrum to isolate each subband, 
down-convert each subband to baseband, 
and then sample each down-converted subband at its Nyquist rate $f^s = B_{m,l} $ using the sampling theorem for bandlimited WSS process.
Up-convert and superpose the reconstructions of all disjoint subbands to recover $A(m,\cdot)$,
which achieves a reconstructed mean square error of zero,
with a total sampling rate equals $\sum_{l=1}^{L}B_{m,l}=B_m$.
Hence the lower bound of the sampling rate for single channel is $B_m$.

For multiple uncorrelated signals $\boldsymbol{\mathit{A}}$, $\boldsymbol{\mathit{R}}_{\boldsymbol{\mathit{A}}} (i,j,\tau) \equiv 0 $ when $i\ne j$, thus the cross power spectral densities satisfy $\boldsymbol{\mathit{S}}_{\boldsymbol{\mathit{A}}}(i,j,\tau) = \mathscr{F}\{\boldsymbol{\mathit{R}}_{\boldsymbol{\mathit{A}}}(i,j,\tau)\}= 0$ when $i\ne j$.
Therefore, the multiple uncorrelated signals are merely a simple superposition of $M$ single-channel signals.

Each single channel
$A(m,\cdot),m\in \{1,...,M\}$ requires sampling density at least $B_m=\mu(\mathcal{B}_m)$.
Summing over all $M$ channels gives a total sampling rate of $\sum_{m=1}^{M}B_m=B$.

Since the joint bandwidth of $\boldsymbol{\mathit{A}}$ is $B$, the total sampling rate requires to reconstruct the multiple uncorrelated signals $\boldsymbol{\mathit{A}}$ is $B$.
We can derive the sampling rates and temporal sampling sets from each subband of $\boldsymbol{\mathit{A}}$ and sample the $N$ channels after mixing using the same temporal sampling set and obtain the corresponding spatial sampling set via the linear transform $\boldsymbol{\mathit{U}}$, thus $\mathcal{S}$ is derived, with total sampling rate $B$.
Therefore, the total sampling density of $\boldsymbol{\mathit{X}}$ satisfies $D(\mathcal{S})\ge \frac{B}{N}$. 
\end{proof}

Theorem 1 asserts that the lower bound of the sampling density for correlated stochastic signals is the total spectral bandwidth of its latent signals divided by the number of correlated signals.




\subsection{Multi-band sampling scheme}
\label{ssec:multi-band sampling scheme}

\begin{figure}[t]
\centering
\includegraphics[width=6.2cm]{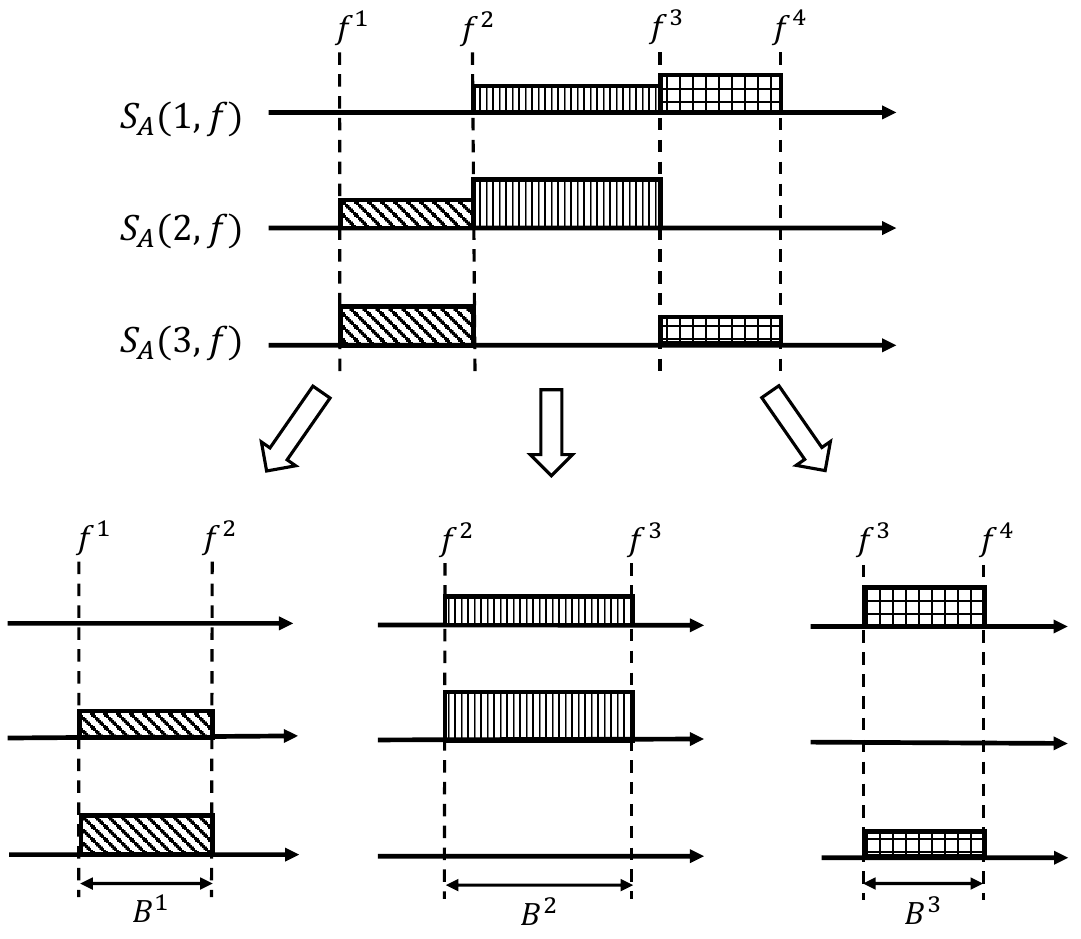}
\caption{Example of subband divison.}
\label{fig:subband division}
\end{figure}

In this subsection, we construct a multi-band sampling scheme that samples $\boldsymbol{\mathit{X}}$ at the lowest total sampling density $\frac{B}{N}$ and achieves zero MSE reconstruction.
The scheme adapts the multi-band partitioning principle, previously explored for graph signals in \cite{sheng2024sampling}, to the context of multiple stochastic signals.

\begin{figure}[t]
\centering
\includegraphics[width=6.6cm]{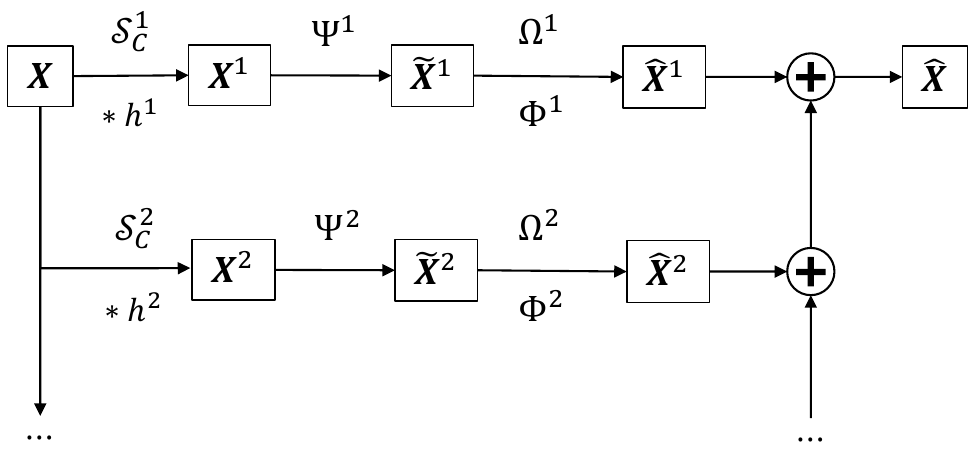}
\caption{Flow chart of the multi-band sampling scheme.}
\label{fig:multiband flow chart}
\end{figure}



We begin by partitioning $\boldsymbol{\mathit{S}}_{\boldsymbol{\mathit{A}}}(\cdot,f)$ into $L$ subbands, each corresponding to a frequency interval $[f^l,f^{l+1}]$. 
The boundaries of the frequency intervals are chosen at frequencies where the PSD of any latent source either starts or vanishes, which ensures that, within each subband, the PSD of a given source is either identically zero or nonzero.
\Cref{fig:subband division} provides an illustrative example of the subband division.

Sampling and reconstruction are performed independently within each subband. 
For subband $l\in \{1,...,L\}$, let $\Gamma^l$ denote the index set of latent sources with nonzero PSD.
The signal components corresponding to the $l$-th subband are extracted by convolving the correlated signals $\boldsymbol{\mathit{X}}$ with a band-pass filter, $h^l(t)$. 
The frequency response of this filter, $H^l(f)$, is equal to unity over $[f^l,f^{l+1}]$ and zero elsewhere. 
This operation yields the following relation for the subband signals: $\boldsymbol{\mathit{X}} \ast h^l= \boldsymbol{\mathit{U}}(\cdot,\Gamma^l)\left (\boldsymbol{\mathit{A}}(\Gamma^l,\cdot) \ast h^l \right )$.


Select $\left | \Gamma^l \right |$ linearly independent rows from $\boldsymbol{\mathit{U}}(\cdot,\Gamma^l)$, where $\Xi^l$ is the index set.
The spatially sampled signals are therefore $\boldsymbol{\mathit{X}}^l=\boldsymbol{\mathit{U}}(\Xi^l,\Gamma^l)\left (\boldsymbol{\mathit{A}}(\Gamma^l,\cdot) \ast h^l \right )$.
For notational convenience set $\boldsymbol{\mathit{U}}^l:=\boldsymbol{\mathit{U}}(\cdot,\Gamma^l)$ and $\boldsymbol{\mathit{U}}_s^l:=\boldsymbol{\mathit{U}}(\Xi^l,\Gamma^l)$.
The corresponding spatial sampling set is $\mathcal{S}_C^l=\{n_{s_1},n_{s_2},...,n_{s_{|\Gamma^l|}}\}$.

$\boldsymbol{\mathit{X}}^l(n,\cdot)$ for $n \in \mathcal{S}_C^l$ is sampled in time at the rate $f^{l,s}=f^{l+1}-f^l=B^l$.
Denote the temporal sampling set by $\mathcal{S}_T^l=\{t_{ns_1},t_{ns_2},...,t_{ns_{B^l}}\}$, we have the multi-band sampled signals at the $l$-th subband: 
\begin{equation*}
\tilde{\boldsymbol{\mathit{X}}}^l(\mathcal{S}_C^l,\mathcal{S}_T^l)=\Psi^l\boldsymbol{\mathit{X}}^l,
\end{equation*}
where $\Psi^l\in \{0,1\}^{B^l\times |\Gamma^l|}$ is the sampling matrix.

In each subband, the sampled signals  $\tilde{\boldsymbol{\mathit{X}}}^l(\mathcal{S}_C^l,\mathcal{S}_T^l)$ can be temporally interpolated by $\Omega^l=\mathrm{sinc}(B^l\allowbreak t^l)e^{j2\pi f^{l,c} t^l}$ to obtain $\hat{\boldsymbol{\mathit{X}}}^l(\mathcal{S}_C^l,\cdot)$,
where $f^{l,c}=\frac{f^l+f^{l+1}}{2}$ and $t^l=t-t_{ns_i}, \, t_{ns_i} \in \mathcal{S}_T^l$.
To lift the reconstruction to all spatial channels, we apply the spatial interpolation operator 
\begin{equation*}
    \Phi^l=\boldsymbol{\mathit{U}}^l ({\boldsymbol{\mathit{U}}_s^l}^H \boldsymbol{\mathit{U}}_s^l)^{-1} {\boldsymbol{\mathit{U}}_s^l}^H,
\end{equation*}
so that the full subband reconstruction is $\hat{\boldsymbol{\mathit{X}}}^l=\Phi^l \hat{\boldsymbol{\mathit{X}}}^l(\mathcal{S}_C^l,\cdot)$, which is valid because $\text{rank}(\boldsymbol{\mathit{U}}_s^l)=|\Gamma^l|$.

Finally, summing the reconstructions over all subbands produces the complete reconstruction 
\begin{equation*}
\hat{\boldsymbol{\mathit{X}}}=\sum_{l=1}^{L}\hat{\boldsymbol{\mathit{X}}}^l,
\end{equation*}
and this scheme attains zero MSE, with $\mathbb{E}\left [ \left |\boldsymbol{\mathit{X}}-\hat{\boldsymbol{\mathit{X}}}\right |^2\right ]=0$.

The process of the proposed scheme is shown in \Cref{fig:multiband flow chart}.

\section{Experiments}
\label{sec:experiments}

In this section, we present experiments on synthetic and real-world data to validate the sampling theory developed in \Cref{sec:sampling theorem} and demonstrate the effectiveness of the proposed multi-band sampling scheme.

\subsection{Results on synthetic dataset}
\label{ssec:results on synthetic dataset}

To synthesize correlated signals, we first generate 
$M$-channel latent signals $\boldsymbol{\mathit{A}}(m,t),m=1,...,M;t=1,...,T$, each a finite-length discrete-time random sequence with length $T$.

The statistical properties of $\boldsymbol{\mathit{A}}$ are defined by prescribing a PSD for each channel. 
$\boldsymbol{\mathit{S}}_{\boldsymbol{\mathit{A}}}(m,\cdot)$ is procedurally generated as a symmetric, piece-wise constant function composed of randomly placed rectangular blocks:
\begin{equation*}
    \boldsymbol{\mathit{S}}_{\boldsymbol{\mathit{A}}}(m,k) = \sum_{i=1}^{K_m} c_{m,i} \cdot \mathbb{I}(k \in \mathcal{F}_{m,i}),
\end{equation*}
where $k$ is the discrete frequency index, $K_m$ is the number of blocks for channel $m$, $c_{m,i}$ are their power levels, and $\mathcal{F}_{m,i}$ are disjoint frequency intervals.
 
Each signal $A(m,t)$ is then synthesized as a discrete, finite-length and real-valued sample function of a WSS process with the prescribed PSD by summing sinusoidal components with amplitudes derived from the PSD and the randomized phases.
For an even signal length $T$, the resulting time-domain expression is:
\begin{equation}
\begin{split}
    A(m,t) ={}& \frac{1}{\sqrt{T}} \Biggl( \sqrt{\boldsymbol{\mathit{S}}_{\boldsymbol{\mathit{A}}}(m,0)} + (-1)^t\sqrt{\boldsymbol{\mathit{S}}_{\boldsymbol{\mathit{A}}}(m,T/2)} \notag \\
    & + \sum_{k=1}^{T/2-1} 2\sqrt{\boldsymbol{\mathit{S}}_{\boldsymbol{\mathit{A}}}(m,k)} \cos\left(\frac{2\pi kt}{T} + \phi_{m,k}\right) \Biggr),
\end{split}
\end{equation}
where $\phi_{m,k}$ are independent random phases drawn uniformly from $[0,2\pi)$. 
Finally, a full-column-rank mixing matrix $\boldsymbol{\mathit{U}}\in \mathbb{R}^{N\times M}$ is randomly generated and the correlated signal samples are obtained using \Cref{eq:signal_model}, producing $N$ correlated channels.



After generating signals, we apply our proposed multi-band sampling scheme, detailed in \Cref{ssec:multi-band sampling scheme}, to sample these signals at the theoretical density $\frac{B}{N}$ and produce the reconstruction $\hat{\boldsymbol{\mathit{X}}}$.

The Normalized Mean Squared Error (NMSE) is adopted to measure the reconstruction quality, defined as:
$\text{NMSE}=\mathbb{E}\left [ \left | \boldsymbol{\mathit{X}}- \hat{\boldsymbol{\mathit{X}}}\right |^2\right ] \Big / {\mathbb{E}\left [ \left | \boldsymbol{\mathit{X}}\right |^2\right ]}$.



\begin{figure}[t]
\centering
\includegraphics[width=3.8cm]{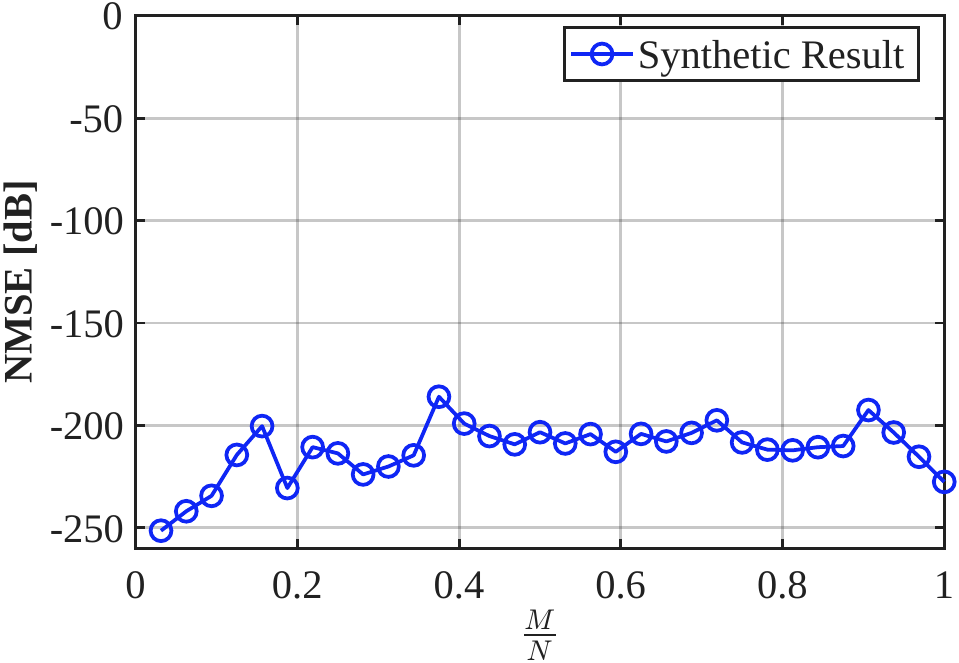}
\caption{Reconstruction NMSE across different $\frac{M}{N}$ ratio.}
\label{fig:synthetic data result}
\end{figure}

First, we evaluate the performance of multi-band sampling scheme under varying degrees of signal correlation, controlled by the ratio $\frac{M}{N}$. 
We fix the number of channels $N$ and vary the ratio $\frac{M}{N}$.
For each configuration, with $T=512$ and $50$ Monte Carlo trials, we compute the NMSE between the synthesized $\boldsymbol{\mathit{X}}$ and its reconstruction $\hat{\boldsymbol{\mathit{X}}}$.
As shown in \Cref{fig:synthetic data result}, the results demonstrate the robustness of our proposed scheme. 
It consistently maintains an NMSE below -180 dB at the theoretical sampling density $\frac{B}{N}$, regardless of the underlying correlation strength (from highly correlated, $\frac{M}{N}\to 0$, to weakly correlated, $\frac{M}{N}\to 1$).




\begin{figure}[t]
  \centering 

  \begin{subfigure}[b]{0.48\linewidth}
    \centering
    \includegraphics[width=3.8cm]{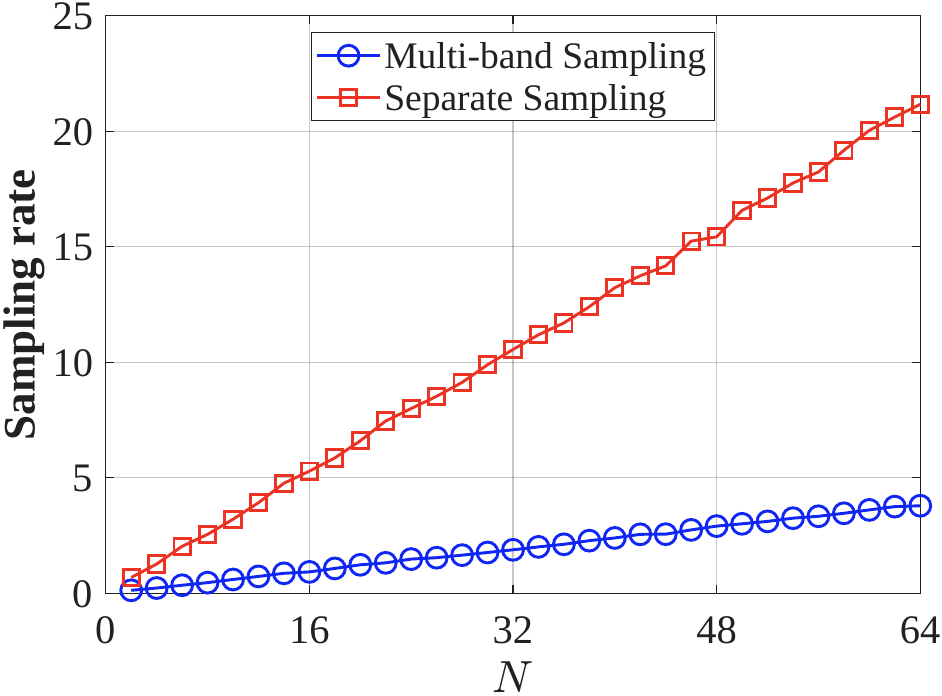}
    \caption{Sampling rate comparison.} 
    \label{fig:sampling_rate}         
  \end{subfigure}
  \hfill 
  \begin{subfigure}[b]{0.48\linewidth}
    \centering
    \includegraphics[width=3.9cm]{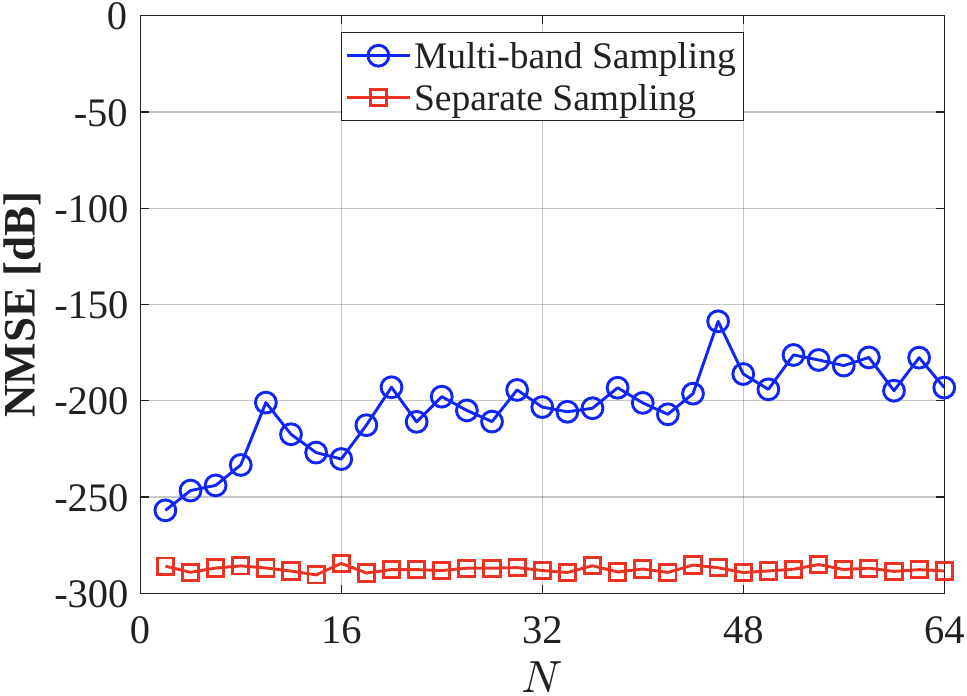}
    \caption{NMSE comparison.} 
    \label{fig:nmse_compare}  
  \end{subfigure}

  \caption{The sampling rate and NMSE by different sampling scheme.}
  \label{fig:multi-band sampling and separate sampling scheme}
\end{figure}



Next, to quantify the sampling efficiency, we compare our multi-band scheme against a separate approach where each channel is sampled independently. 
We fix the ratio $\frac{M}{N}=0.5$ and set $T=512$, with results averaged over $50$ Monte Carlo trials.
The results are summarized in \Cref{fig:multi-band sampling and separate sampling scheme}. 
\Cref{fig:sampling_rate} compares the sampling rates, which is defined as the total number of samples divided by $T$. 
It is evident that our proposed scheme requires a substantially lower sampling rate than the baseline approach. 
Crucially, \Cref{fig:nmse_compare} confirms this efficiency gain does not compromise reconstruction fidelity.
Despite the significant reduction in samples, the NMSE of our method remains at a comparably low level to the separate sampling scheme, confirming near-lossless performance. 
Taken together, these results validate that our method significantly reduces sampling costs while maintaining reconstruction accuracy.

\subsection{Results on real dataset}
\label{ssec:results on real dataset}




\begin{figure}[t]
  \centering 

  \begin{subfigure}[b]{0.44\linewidth}
    \centering
    \includegraphics[width=3.6cm]{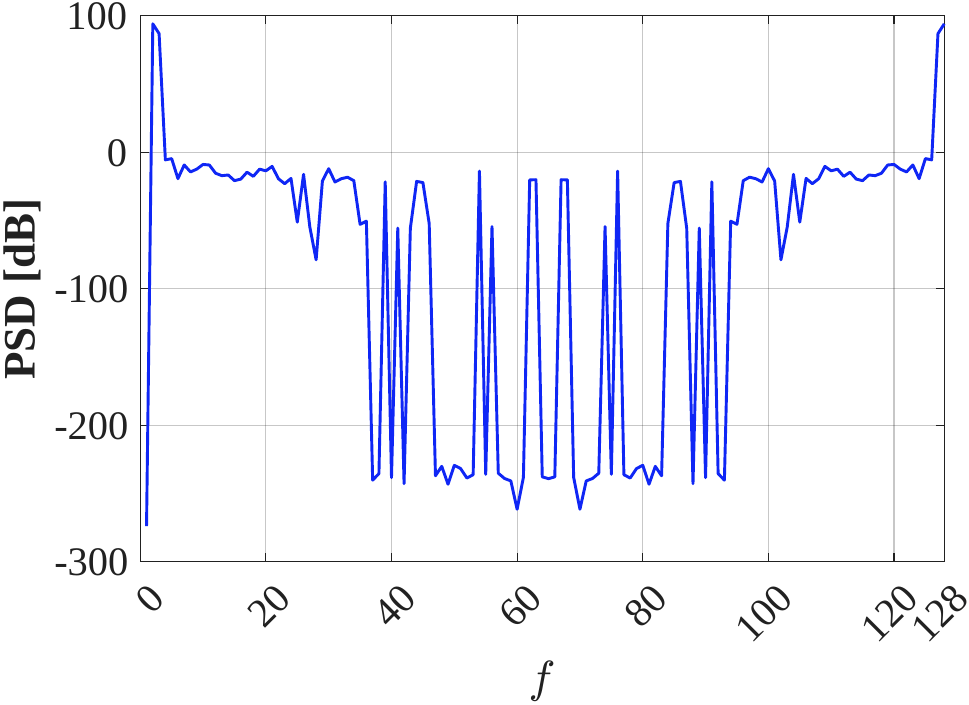}
    \caption{Estimated PSD.} 
    \label{fig:real_PSD}         
  \end{subfigure}
  \hfill 
  \begin{subfigure}[b]{0.55\linewidth}
    \centering
    \includegraphics[width=4.4cm]{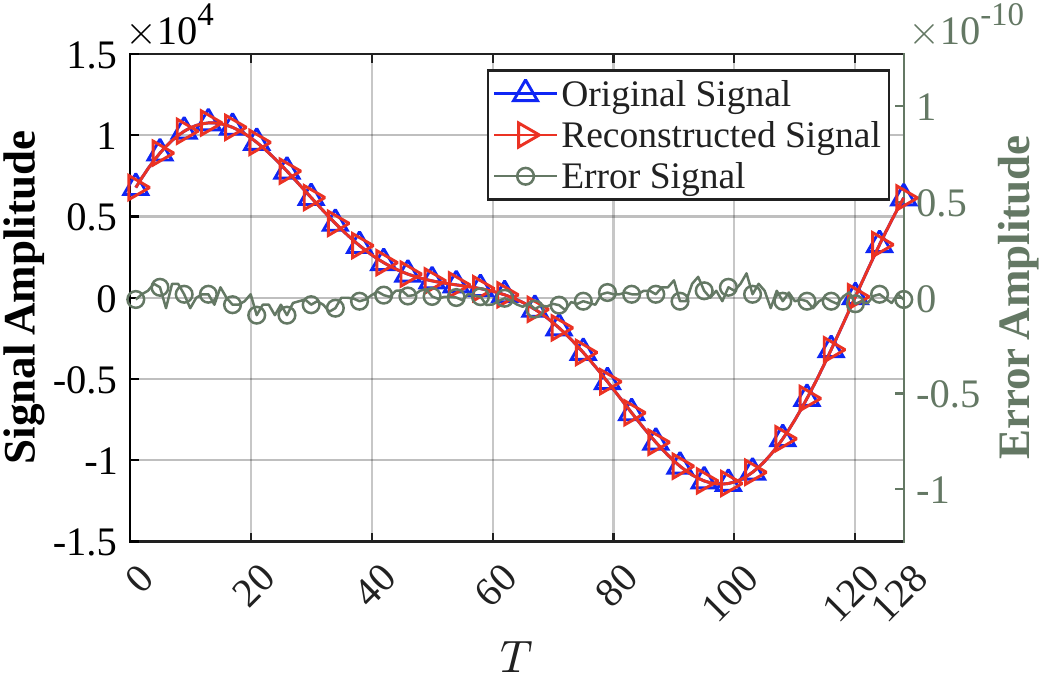}
    \caption{Time-domain reconstruction.} 
    \label{fig:real_reconstruction}  
  \end{subfigure}

\caption{Real data reconstruction of a representative channel.}
\label{fig:real data result}
\end{figure}

We test the multi-band sampling scheme on the public  dataset Time Series Database Library (TSDL) \cite{tsdl}, which contains $648$ time series spanning domains such as finance, agriculture, meteorology, physics, production and sales.

We select a subset of series that satisfy the WSS assumption and truncate every series to a fixed length $T$ for comparability. 
A collection of $N$ series from different sources forms the raw correlated signals.
We project the raw correlated signals into the latent uncorrelated component space via PCA to estimate the mixing matrix and the unprocessed latent signals.
For each estimated source we compute an empirical PSD.

Finite-sample effects and observational noise may induce spectral leakage and estimation noise.
To mitigate spurious weak spectral components, we threshold each source's estimated PSD by setting values below 5\% of that source's maximum PSD to zero.
The thresholded spectra define the spectral supports used in subsequent processing and determine the total spectral bandwidth $B$.
After preprocessing, we derive the thresholded latent signals $\boldsymbol{\mathit{A}}$ and corresponding correlated signals $\boldsymbol{\mathit{X}}$.
Similar to the synthetic data construction, we then apply our multi-band sampling scheme at the resulting theoretical density $\frac{B}{N}$ to reconstruct the correlated signals.


Set $T=128$ and $N=10$, the results for a representative channel of multiple correlated signals $\boldsymbol{\mathit{X}}$ are presented in \Cref{fig:real data result}. 
\Cref{fig:real_PSD} displays the estimated PSD of this single channel.
As illustrated in \Cref{fig:real_reconstruction}, the reconstructed signal closely tracks the original signal waveform.
Visually, the two waveforms are nearly indistinguishable, and the error signal plotted below remains close to zero with a significantly smaller amplitude relative to the signal itself, providing a clear validation of our scheme's high-fidelity performance.

Quantitatively, the reconstruction achieves a total NMSE of -218.6915 dB between the observed signals $\boldsymbol{\mathit{X}}$ and the complete reconstruction $\hat{\boldsymbol{\mathit{X}}}$, confirming the effectiveness of our proposed method on real-world data.




\section{Conclusion}
\label{sec:conclusion}

In this paper, based on a latent source model, we establish a sampling theory for multiple correlated stochastic signals. 
We derive a fundamental lower bound on the total sampling density required for reconstruction with zero MSE, proving it equals the total spectral bandwidth of the latent sources divided by the number of correlated signals. 
Furthermore, we propose a constructive multi-band sampling scheme to achieve this bound via spectral partitioning. 
Experiments on synthetic and real datasets confirm that our proposed scheme achieves near-lossless reconstruction at the theoretical minimum rate, establishing a direct link between inter-channel correlation and sampling efficiency.

\vfill\pagebreak

\bibliographystyle{IEEEtran}
\bibliography{refs}

\end{document}